\documentclass[times,10pt,twocolumn]{article}

\usepackage{latex8}
\usepackage{times}
\usepackage{amsmath,amssymb,amsthm}
\usepackage{multirow}
\usepackage{dsfont}
\usepackage{eurosym}
\usepackage{url}
\usepackage{bm}
\usepackage{authblk}
\usepackage{pgf}
\usepackage{tikz}
\usetikzlibrary{arrows,automata}
\DeclareGraphicsExtensions{.ps,.eps,.pdf}

\newcommand{\D}{\mathcal{D}}
\newcommand{\U}{\mathbf{U}}
\newcommand{\Snc}{\mathbf{S}}
\newcommand{\T}{\mathbf{T}}
\newcommand{\R}{\mathbf{R}}
\newcommand{\N}{\mathbb{N}}
\newcommand{\Z}{\mathbb{Z}}
\newcommand{\X}{\mathbf{X}}
\newcommand{\Y}{\mathbf{Y}}
\newcommand{\Zed}{\mathbf{Z}}
\newcommand{\iFF}{\Leftrightarrow}

\newtheorem{theorem}{Theorem}

\newtheorem{corollary}[theorem]{Corollary}
\newtheorem{definition}[theorem]{Definition}

\begin{document}

\title{Bounded Reachability for Temporal Logic over Constraint Systems}

\author[1]{Marcello M. Bersani}
\affil[1]{Politecnico di Milano\\
Milano, Italy\\
\{bersani,frigeri,morzenti,rossi,sanpietro\}@elet.polimi.it}
\author[1]{Achille Frigeri}
\author[1]{Angelo Morzenti}
\author[2]{\\Matteo Pradella}
\affil[2]{CNR IEIIT-MI\\
Milano, Italy\\
pradella@elet.polimi.it}
\author[1]{Matteo Rossi}
\author[1]{Pierluigi San Pietro}

\maketitle
\thispagestyle{empty}

\begin{abstract}
This paper defines CLTLB($\D$), an extension of PLTLB (PLTL with both past and future
operators) augmented with atomic formulae built over a constraint
system $\D$.
The paper introduces suitable
restrictions and assumptions that 
make the satisfiability problem decidable in many cases, although the problem is undecidable in the
general case.
Decidability is shown for a large class of constraint systems, and an
encoding into Boolean logic is defined. This paves the way for applying existing SMT-solvers for 
checking the Bounded Reachability problem, as shown by various experimental results.
\end{abstract}


\section{Introduction}
Many extensions of temporal logic or automata have been proposed
with the goal of verifying infinite-state systems.
Among the many extensions of Propositional Linear Temporal Logic (PLTL),
there have been proposals of allowing formulae which may include arithmetic
constraints belonging to a specific constraint system \cite{CC00,DD02} .
These logics are well-suited to define properties of
infinite-state systems, but, 
unfortunately for the aim of automatic verification, previous results have also shown 
the undecidability of the satisfiability problem, at least in the general case \cite{DG06}.
Here we define a more general logic, called CLTLB($\D$), 
which is an extension of PLTLB (PLTL with Both future and past
operators), allowing arithmetic
constraints belonging to a generic constraint system.
To cope with undecidability, already known for the less general case above, we introduce suitable assumptions
concerning the structure of models, but without 
any syntactic restriction on formulae.
Models only consider partial valuations of
arithmetic variables: the satisfiability of CLTLB($\D$) then turns to be decidable,
provided that the constraint system $\D$ has a decidable decision procedure.
We then define the Bounded Reachability Problem (BRP) for CLTLB($\D$),
which can be decided by showing its equivalence 
to the satisfiability of
CLTLB($\D$) over partial valuations. 
We realized a Bounded Reachability Checker by using SMT-solvers natively implementing 
decision procedures for Quantifier-Free Integer Difference Logic with Uninterpreted Functions (QF-UFIDL)
and Quantifier-Free Linear Integer Arithmetic with Uninterpreted Functions (QF-UFLIA).
Experimental results using the Zot toolkit \cite{PMS07,PMS09} 
show that, the greater expressiveness of CLTLB($\D$) notwithstanding, 
the encoding of the propositional part is considerably faster and with smaller
memory footprint than existing encodings of PLTL based on SAT.

The paper is structured as follows.
Section \ref{sec:undec} relates on the state of the art in extending PLTL with constraint systems.
Section \ref{sec:cltlb} introduces CLTLB($D$), while Section \ref{sec:dec} presents various 
decidability and undecidability results. 
Section \ref{sec:brp} introduces and solves the BRP.
Section \ref{sec:encoding} defines an encoding of CLTLB($\D$) into logics suitable for SMT-based verification.
Section \ref{sec:tests} relates on the performance of experimental results of the original SAT-based plugins of Zot
with the SMT-based ones on a number of examples taken from different application domains.
Finally, Section~\ref{sec:conclusions} draws a few conclusions and outlines future research. 

\section{State of the art}\label{sec:undec}

Among the various proposals of extension of LTL, 
CLTL (Counter LTL) has been defined in \cite{CC00}.
CLTL is, essentially, Propositional LTL with future
operators (PLTL), with in addition terms that are
arithmetic constraints in Integer Difference
Logic (DL).
However, by reducing the recurrence
problem for Minsky machines to the satisfiability of a CLTL
formula, it is shown that the logic is undecidable, hence unsuitable for automatic verification.

A generalization of CLTL is 
CLTL($\D$) \cite{DD02}, where the arithmetic constraints belong to a general
constraint system $\D$.
If $\D$ has an intrinsic counting mechanism, i.e., it
contains equality and a binary injective relation $R$ such that its graph is a DAG,
then CLTL($\D$) is undecidable. 
Indeed, a relation satisfying the hypothesis of
the theorem generalizes the ``successor'' function and
can be used to define constraints of the form $y = z + 1$.
\cite{DG06} proves the undecidability of the
satisfiability problem of CLTL$_m^l$(DL), which is the restriction of
CLTL(DL) to formulae with at most $m$ variables and of depth
less or equal to $l$.
CLTL$_m^l$(DL) is shown to be 
$\Sigma_1^1$-hard for $m>1$
and $l>1$, while CLTL$_1^1$(DL) is shown to be \textsc{PSPACE}-complete.

For practical model-checking, a large variety of infinite-state systems can be
effectively represented by counters systems.
In this case, interesting results on verifying safety and reachability
properties can be obtained by constraining the control graph of the counters
system to be flat \cite{CJ98, B98}, i.e., no control state occours in more
than one simple cycle.
Properties are defined by means of Presburger arithmetic constraints
but they are not considered in the framework of any temporal logic, for
instance, like CLTL or CLTL($\D$) described above.
In \cite{DFGD06}, authors extend some results about flat systems to more general classes
of infinite-state systems in which some first-order extensions of CTL$^\star$
have decidable model-checking. 

To cope with undecidability,  \cite{dMRS02} 
describes a reduction of infinite BMC to a satisfiability problem of Boolean
constraints formulae.
By translating LTL formulae into a corresponding B{\"u}chi
automaton, a BMC problem is reduced to the satisfiability of a mixed arithmetic-Boolean
formula.
The authors also give a proof of soundness and completeness for the $\U$-free fragment of the
logic.
In this case, the BMC problem is solved by means of a loop-free encoding, since
$\U$-free formulae
can always be  translated into an automaton over \emph{finite} words accepting
a prefix of all infinite paths which satisfy it.
In all other cases, generic LTL formulae are translated into
a corresponding B\"uchi automaton with acceptance conditions involving
an implicit periodicity constraint over counters.
However, this translation does not work when counters do not behave periodically.
For instance, consider a transition system defining a non-periodic, strictly-increasing
counter $x$ starting at $0$.
Property $\top\U (x<0)$ does not hold for this system,
but the B{\"u}chi automaton corresponding to its negation imposes a periodic constraint over the sequence of values of $x$, which cannot be satisfied.
Hence, using the translation outlined above, verification of formula $\top\U (x<0)$ for the strictly-increasing counter improperly yields true.

We define a complementary, purely descriptive, approach which solves this problem.
It is also aimed at solving reachability problems for infinite-state systems
whose propositional, possibly periodic, behaviors induce a finite prefix of values of variables and
satisfying a CLTLB($\D$) specification, instead of LTL properties just
over arithmetic constraints. 

\section{A Temporal Logic over Constraint Systems}\label{sec:cltlb}

This section presents an extension to Kamp's \cite{Kam68} PLTLB,
by allowing formulae over a constraint system.
As suggested in \cite{CC00}, and unlike the approach of \cite{D04}, the propositional variables
of this logic are 
Boolean terms or atomic arithmetic constraints.

Let $V$ be a set of variables; a {\em constraint system} is a pair
$\D=\langle D, \Pi\rangle$ where $D$ is a specific domain of interpretation for variables and
constants and
$\Pi$ is a family of relations on elements of $D$.
An {\em atomic $\D$-constraint} is a term of the form $R^n(x_1,\dots,x_n)$, where $R^n$
is an $n$-ary relation on $D$ and $x_1, \dots, x_n$ are variables.
A $\D$-valuation is a mapping $v: V \to D$, i.e., an assignment of a value in $D$
to each variable.
A constraint is {\em satisfied} by a $D$-valuation $v$,
written $v \models R(x_1,\dots,x_n)$, if $\left(v(x_1),\dots,v(x_n) \right)
\in R$.

Let $AP$ be a set
of atomic propositions and $\D=\langle D, \Pi\rangle$ a constraint system. CLTLB($\D$) is
defined as an
extension of PLTLB,
by combining Boolean atoms with arithmetic temporal terms defined in $\D$.
The resulting logic is actually equivalent to the quantifier-free fragment of
FOLTL \cite{E90} over signature $\{\Pi, AP\}$.
The syntax of CLTLB($\D$) is defined as follows:
\begin{equation*}
\begin{gathered}
  \phi :=
  \left\{
  \begin{gathered}
    p \mid R(\varphi_1, \dots, \varphi_n) \mid \phi \wedge \phi \mid \neg \phi \mid \\
   \X\phi \mid \Y\phi 
\mid \phi\U\phi \mid \phi\Snc\phi
  \end{gathered}
  \right.\\
  \varphi := x \mid \X \varphi \mid \Y \varphi
\end{gathered}
\end{equation*}
where $p \in AP$, $x \in V$, $\X$ and $\Y$  are the usual ``next'' and  ``previous'' operators,
$\U$ and $\Snc$ are the usual ``until'' and ``since''
operators, $R \in \Pi$, $\X^j$ and $\Y^j$ are shorthands for
$j$ applications of $\X$ and $\Y$ (e.g.,
$\X^2 \equiv \X\X$). 
Each formula $\varphi$ is called an \emph{arithmetic temporal term} (a.t.t.).
Its \emph{depth} $|\varphi|$ is
the total amount of temporal shift needed in evaluating $\varphi$: 
\begin{equation*}
\begin{gathered}
|x| = 0,\\
|X(\varphi)| = |\varphi| + 1,\\
|Y(\varphi)| = |\varphi| - 1.
\end{gathered}
\end{equation*}

Let $\phi$ be a CLTLB($\D$) formula, $x$ a variable and $\Gamma_x$ the set of all
a.t.t.'s occurring in $\phi$ in which $x$ appears. We define the ``look-forwards''
$\lceil \phi \rceil_x$ and
``look-backwards''
$\lfloor \phi \rfloor_x$ of $\phi$ relatively to $x$ as:
\[
\begin{gathered}
\lceil \phi \rceil_x = \max_{{\varphi_i} \in \Gamma_x}\{0, |\varphi_i|\}\\
\lfloor \phi \rfloor_x = \min_{{\varphi_i} \in \Gamma_x}\{0, |\varphi_i|\}
\end{gathered}
\]
The above definitions may naturally be extended to the set $V$ of all variables (by letting
$\lceil \phi \rceil = \max_{x\in
V} \{\lceil \phi \rceil_x\}$,
$\lfloor \phi \rfloor_x =  \min_{x\in V}\{\lfloor \phi \rfloor_x\}$).
Hence, $\lceil \phi \rceil$ ($\lfloor \phi
\rfloor$) is the largest (smallest) depth of all the a.t.t.'s of $\phi$, representing
the length of the future (past) segment needed to evaluate 
$\phi$ in the current instant.

The semantics of a formula $\phi$ of CLTLB($\D$) is defined w.r.t. a
linear time structure $\pi_{\sigma}=(S, s_0, I, \pi, \sigma, L)$, where $S$ is a set of
states, $s_0$ is the initial state,
$I:\{j \mid \lfloor \phi\rfloor \le j \le -1\}
\times V\to D$ is an assignment,
$\pi\in s_0 S^\omega$ is an \emph{infinite path}, 
$\sigma: \N \times V \to D$ is a sequence of $\D$-valuations
and $L: S \to 2^{AP}$ is a labeling function.
From now on, the set of all sequences of $\D$-valuations is denoted by
$\Sigma$.
Function $I$ defines the valuation of variables for each time instant in
$\{j \mid \lfloor \phi\rfloor \le j \le -1\}$, i.e., for time instants before 0; this way $\sigma$ can be extended to a.t.t.'s.
Indeed, if $\varphi$ is an a.t.t., $x$ is the variable in $\varphi$,
$i\in \mathbb{N}$
and $\sigma^i(x)$ is a shorthand for $\sigma(i,x)$, then:
\begin{equation*}
\sigma^i(\varphi) = \left\{
                      \begin{array}{ll}
                        \sigma^{i+|\varphi|}(x), & \hbox{if $i+|\varphi|\geq 0$;} \\
                        I(i+|\varphi|,x), & \hbox{if $i+|\varphi|<0$.}
                      \end{array}
                    \right.
\end{equation*}
The semantics of a CLTLB($\D$) formula $\phi$ at instant $i\in\N$
over a linear structure $\pi_{\sigma}$ is recursively defined by means of a satisfaction relation
$\models$ as follows, for every formulae $\phi, \psi$ and
for every a.t.t. $\varphi$:
\begin{equation*}
\begin{aligned}
\pi_\sigma^i \models p  &\iFF  p \in L(s_i) \text{ for } p \in AP \\
\pi_\sigma^i \models R(\varphi_1, \dots, \varphi_n) &\iFF \\
  (\sigma^{i+|\varphi_1|}(x_{\varphi_1}), &\dots, \sigma^{i+|\varphi_n|}(x_{\varphi_n}) ) \in R\\
\pi_\sigma^i \models \neg p &\iFF  \pi_\sigma^{i} \not\models p \\
\pi_\sigma^i \models \phi \wedge \psi &\iFF  \pi_\sigma^i \models \phi
\, \text{and} \, \pi_\sigma^i \models \psi\\
\pi_\sigma^i \models \X \phi &\iFF \pi_\sigma^{i+1} \models \phi \\
\pi_\sigma^i \models \Y \phi &\iFF \pi_\sigma^{i-1} \models
\phi \wedge i>0\\
\pi_\sigma^i \models \phi\U\psi &\iFF
\left\{ \begin{gathered}
  \exists \, j\geq i: \pi_\sigma^j \models \psi \, \wedge \\ \pi_\sigma^n
\models \phi \ \ \forall\, i\leq n < j
\end{gathered}\right. \\
\pi_\sigma^i \models \phi\Snc\psi &\iFF
\left\{\begin{gathered}
\exists \, 0\leq j \leq i: \pi_\sigma^j \models \psi \, \wedge \\
\pi_\sigma^n \models \phi \ \ \forall\, j < n \leq i
\end{gathered}\right. \\
\end{aligned}
\end{equation*}
where $x_{\varphi_i}$ is the variable that appears in $\varphi_i$.
The semantics of $\phi$ is well defined, as any valuation
$\sigma^i$ is defined for all $i\geq \lfloor \phi \rfloor$, because of assignment $I$.
A formula $\phi \in$ CLTLB($\D$) is \emph{satisfiable} if there exists a linear
time structure $\pi_{\sigma}=(S, s_0, I,\pi, \sigma, L)$ such that $\pi_\sigma^0
\models
\phi$ (in which case $\pi_{\sigma}$ is a \emph{model} of $\phi$).
Without loss of generality, one may assume that all formulae are in \emph{positive normal
form},
where negation may only occur in front of atomic constraints.
In fact, by introducing as primitive the connective $\vee$, the dual operators
``release'' $\R$, ``trigger'' $\T$ and ``previous'' $\Zed$
defined as: 
$\phi\R \psi
\equiv \neg(\neg \phi \U \neg\psi)$, $\phi\T \psi
\equiv \neg(\neg \phi \Snc \neg\psi)$ and $\Zed \phi \equiv \neg \Y \neg \phi$, and
by applying De Morgan's rules, every CLTLB formula can be rewritten
into positive normal form.
\section{(Un)decidability of CLTLB($\D$)}\label{sec:dec}
As a first result, by exploiting well-know properties of PLTLB, we prove the equivalence of CLTLB($\D$) to
CLTL($\D$) for a quantifier-free constraint system $\D$,
w.r.t. \emph{initial} equivalence.
Then, as a corollary of results described in
Section \ref{sec:undec}, we obtain the undecidability of
CLTLB($\mathcal{D}$) for a large class of constraint systems.

In the following, as customary, we denote with $\pi$ a structure for a PLTLB formula.  
\begin{definition}\label{def:gl/in-eq}
Two PLTLB formulae $\phi, \psi$ are \emph{globally} equivalent, written
$\phi \equiv_g \psi$, if for all linear-time structures $\pi$ it is $\pi^i \models \phi \iFF \pi^i \models
\psi$ 
for all $i \in \N$.
Two PLTLB formulae $\phi, \psi$ are \emph{initially} equivalent, written
$\phi \equiv_i \psi$, when $\pi^0 \models \phi \iFF \pi^0 \models
\psi$ for all linear-time structures $\pi$.
\end{definition}
In \cite{GPSS80} it is shown that any PLTLB formula is initially
equivalent to a PLTL formula, while the two logics are not globally
equivalent (see \cite{Sch02} for details).
In order to extend this result to the constrained case, we need to
introduce new temporal operators.
CLTLB($\D$), as defined in Section \ref{sec:cltlb},
includes the ``non-strict'' until (resp. since) operator, in which formula
$\phi \U \psi$ (resp. $\phi \Snc \psi$) holds in an instant $i$ when $\psi$ holds in $i$, and only if $\phi$ holds starting from $i$.
The ``strict'' version of until
$\U^>$, instead, does not require this:
\begin{equation*}
\pi_\sigma^i \models \phi\,\U^>\psi \iFF
\left\{ \begin{gathered}
  \exists \, j > i: \pi_\sigma^j \models \psi \, \wedge \\ \pi_\sigma^n
\models \phi \ \ \forall\, i < n < j
\end{gathered}\right.
\end{equation*}
and similarly for the strict since $\Snc^>$.
It is well known that the following global equivalences hold for any $\phi$,$\psi$:
\begin{equation*}
\begin{array}{cc}
\X \phi \equiv_g \perp \U^> \phi, & \phi \U \psi \equiv_g \psi \vee ( \phi \wedge
\phi \U^> \psi); \\
\Y \phi \equiv_g \perp \Snc^> \phi, & \phi \Snc \psi \equiv_g \psi
\vee (\phi
\wedge \phi \Snc^> \psi). \\
\end{array}
\end{equation*}
Using the previous equivalences, Gabbay \cite{G87} proved that any PLTLB
formula is globally equivalent to a separated PLTLB formula, i.e. a
Boolean combination of formulae containing either $\U^>$ ($\U^>$-formulae) or
$\Snc^>$ ($\Snc^>$-formulae), but not both.
Since this theorem preserves all semantic properties, i.e., it is actually
a rewriting syntactic procedure over formulae, it extends also to the
case of CLTLB($\D$), provided that each arithmetic constraint is
accounted as a propositional letter.
In particular, a.t.t.'s $\X x$/$\Y x$ are
not rewritten using strict-until/-since operators, but are considered
as is, since their semantics depends on the underlying sequence $\sigma$ as
defined before.
Then, we need to show that $\Snc^>$-formulae can be translated
into \emph{initially} equivalent $\U^>$-formulae.
More precisely, we prove the following:
\begin{theorem} \label{thm: init_glob_equiv}
Any CLTLB($\D$) formula is initially
equivalent to a CLTL($\D$) formula, while the two logics are not globally
equivalent.
\end{theorem}
\begin{proof}[Proof sketch.]
We first prove that CLTL($\D$) is not globally equivalent to CLTLB($\D$) by
providing a counterexample. Formula $\top
\Snc A$, where $A\in AP$, was shown in \cite{E90} to have no
globally equivalent PLTL formula. Now, suppose $\phi$ is a CLTL($\D$)
formula globally equivalent to CLTLB($\D$) formula $\top \Snc A$. Then, for the above reason, it
should constrain
at least one of its arithmetic variables, by a non trivial arithmetic
formula. Since $\top
\Snc A$ does not constrain any arithmetic variables, some of its models cannot be models of
$\phi$.


To prove the initial equivalence we suppose each formula is written
using only $\U^>$ and $\Snc^>$ operators, using the equivalences above.
From Gabbay's Separation Theorem such a formula can be rewritten
to a separated CLTLB($\D$) formula which is a Boolean combination of
$\Snc^>$- and $\U^>$-formulae.
The proof is concluded by noticing that any $\Snc^>$-formula
is trivially initially equivalent to false.
\end{proof}
\begin{corollary}
Let $\D=\langle D,\Pi\rangle$ be a constraint system where $\Pi$ contains equality and a binary
relation $R$ such that $(D,R)$ is a DAG; then, satisfiability of CLTLB($\D$) is undecidable.
\end{corollary}
In the following, in the case of a decidable constraint system $\D$, we prove the decidability of the satisfiability and the
model checking problems for CLTLB($\D$) formulae for
partial $\D$-valuations, in which that for all computations
the value of counters will be considered only for a fixed number of
steps.
The counting mechanism of $\D$ is not altered along finite paths by
means of constraints imposing periodicity of values of variables and all relations
are still considered over infinite, possibly periodic, paths.
This allows us to define a
complementary approach to the one of \cite{dMRS02}, aimed at bounded satisfiability checking
\cite{PMP08} 
and BMC of
infinite-state systems.
With this assumption, any periodic behavior which induces a
finite, even periodic, prefix of values of variables ruled by the counting mechanism and
satisfying a CLTLB($\D$) formula, can be represented.
An arithmetic variable varying over a bounded set may still be
represented by its Boolean representation and be part of the
propositional infinite paths.
It is worth noticing that, since we limit the counting mechanism along
finite paths, the partial model is an under-approximation, due to the intrinsic
undecidability of the general problem.
\begin{definition}
Let $\phi$ be a CLTLB($\D$) formula and $k \in \N$, then a
\emph{k-partial} $\D$-valuation $\sigma_k$ for $\phi$ is a relation in
$\{i \in \Z \mid i \geq \lfloor \phi\rfloor \} \times V \times D$ with the condition
that for each variable $x$ occurring in $\phi$, its restriction over
$\{i \in \Z \mid \lfloor \phi \rfloor_x \leq i \leq k + \lceil
\phi\rceil_x\}\times \{x\} \times D$
is a function from
$\{i \in \Z \mid \lfloor \phi \rfloor_x \leq i \leq k + \lceil
\phi\rceil_x\}\times \{x\}$ to $D$. Then, $\Sigma_k$ is the set of all $k$-partial $\D$-valuations for $\phi$.
\end{definition}
Informally, $\sigma_k$ defines a unique value for each
counter $x$ from $0$ up to the bound $k$ by means of boundaries
conditions in the intervals $\{i\in\Z\mid \lfloor \phi \rfloor_x\leq i< 0\}$ and $\{i\in\Z\mid k<i\leq k+\lceil \phi \rceil_x\}$, and it accounts for relations over infinite, even periodic, paths, after $k$.
For the case of $k$-partial $\D$-valuation one can define a
semantics of CLTLB($\D$) formulae.
It coincides with the semantics of the
(full) $\D$-valuations except for the case of arithmetic relations
$R$; namely:
\begin{equation}\label{eq:rel-over-partial}
\begin{array}{c}
\pi_{\sigma_k}^i \models R(\varphi_1, \dots, \varphi_n) \iFF \\
\forall y_1, \dots, y_n \in D \text{ s.t. } \forall 1\leq j\leq n, 
(i+|\varphi_j|, x_{\varphi_j}, y_j) \in \sigma_k \\
\text{then } (y_1, \dots, y_n) \in R,
\end{array}
\end{equation}
where $x_{\varphi_j}$ is the variable that appears in $\varphi_j$.
If $\sigma_k$ is a function, this semantics reduces
exactly to the previous one.
The satisfiability problem for a CLTLB($\D$) formula $\phi$ over
$k$-\emph{partial} $\D$-valuations is that of looking for a (partial) linear
time structure $\pi_{\sigma_k}=(S, s_0, \pi,\sigma_k, L)$ such that $\pi_{\sigma_k}^0 \models
\phi$.
It is worth noticing that the initialization function $I$ is implicit in the definition of $\sigma_k$.
\begin{theorem}
The satisfiability of a CLTLB($\D$) formula $\phi$ over
$k$-partial $\D$-valuations is decidable when $\D$ is decidable.
\end{theorem}
\begin{proof}[Proof sketch.]
Thanks to the initial equivalence of CLTLB($\D$) and CLTLF($\D$)
formulae (Theorem \ref{thm: init_glob_equiv}), we assume 
without loss of generality that
$\phi \in$ CLTLF($\D$); moreover, we assume that a.t.t.'s do not appear negated (i.e., negated a.t.t.'s are transformed into the positive form of the complement relation) and that constraints in
$\phi$ are in disjunctive normal form (i.e., disjunction of
conjunction of propositions and a.t.t.'s). 
Let $\mathcal{C}$ be the set containing all conjunctions of such terms, and let $\mathcal{A}_\phi$ be the
corresponding B\"uchi
automaton whose alphabet is $A=\mathcal{P}(\mathcal{C})$.
The satisfiability of $\phi$ is reduced to the emptiness of
$L(\mathcal{A}_{\phi})$.
In fact, if $L(\mathcal{A}_{\phi})$ is empty, then $\phi$ is unsatisfiable.
If $L(\mathcal{A}_{\phi})$ is not empty, then $\mathcal{A}_{\phi}$ has one or 
more strongly connected components that are reachable from an
initial state and contain a final state. Hence, it is enough to check if there exists a path
of length $k$ from the initial state (which also considers the initial values of the
variables) that can be extended to one of the above components and which satisfies each constraint.
This is decidable, because the consistency problem of $\D$ is decidable.
Finally, it can be shown that the finite sequence of variable assignments appearing in such a path of length $k$ can be extended to a $k$-partial $\D$-valuation on which $\phi$ is satisfied, for example by using the empty relation outside those instants in which the valuation is required to be a function.
\end{proof}

Section \ref{sec:encoding} computes an estimation of the complexity
of problem for a large class of constraint system.

 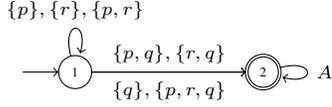
\begin{figure}\centering
\begin{tikzpicture}[node distance=5cm,auto,
                    every state/.style={draw=black!100,thin}]

  \node[state, initial by arrow, initial text={}, scale=0.5] (q_0) {1};
  \node[state, accepting, scale=0.5] (q_1) [right of=q_0] {2};

  \path[->] (q_0) edge [loop above] node {{\scriptsize $\{p\},\{r\},\{p,r\}$}} (q_0)
            (q_0) edge [] node[above] {\scriptsize{$\{p,q\},\{r,q\}$}} (q_1)
            (q_0) edge [] node[below] {\scriptsize{$\{q\}, \{p,r,q\}$}} (q_1)
            (q_1) edge [loop right] node {{\scriptsize $A$}} (q_1);

\end{tikzpicture}
\caption{B\"uchi automaton for $(p\vee r)\U q$, with $p := x = \Y y + 1$, $r:= y = x+2$ and $q := y \leq \X^2 x \wedge x < \X x$.}
\label{buchiA_fig}
\end{figure}
As an illustrative example, consider the satisfiability of the formula $\phi :=
(p\vee r) \U q$ where $p :=x = \Y y + 1$, $r:= y = x + 2$ and $q := y \leq
\X^2 x \wedge x < \X x$ and let be $k = 3$.
The emptiness problem reduces to finding a consistent
assignment to $x$ and $y$ along a path of length $3$ over the
B\"uchi automaton $\mathcal{A}_\phi$ on the alphabet $A =
\{\varnothing,\{p\},\{q\},\{r\},\{p,r\},\{p,q\},\{r,q\},\{p,q,r\}\}$ shown in
Fig. \ref{buchiA_fig}.
Actually, we need to check
the consistency for at least one prefix of length 3 of $L(\mathcal{A}_\phi)$.
In Fig. \ref{constr_fig} we show the corresponding graph of the
constraints to be solved for the word $\{p,r,q\}\{p\}\{p,q\}$.
A dashed line means that the constraint in the label does not hold,
numbers in the circles are possible assignments to the variables, while a
blank means that the corresponding value is irrelevant, and can be left undefined.
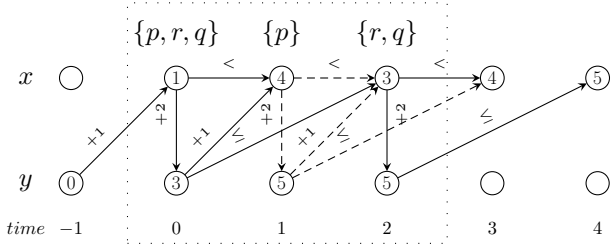
\begin{figure}\centering
\begin{tikzpicture}[node distance=4cm,auto,
                    every state/.style={draw=black!100,scale=0.35}, >=stealth]

  \draw [loosely dotted] (0.75,-2.2) rectangle (5,1);

  \node[state] (x_!) {};
  \node[state] (x_0) [right of=x_!] {{\huge $1$}};
  \node[state] (x_1) [right of=x_0] {{\huge $4$}};
  \node[state] (x_2) [right of=x_1] {{\huge $3$}};
  \node[state] (x_3) [right of=x_2] {{\huge $4$}};
  \node[state] (x_4) [right of=x_3] {{\huge $5$}};

  \node[state] (y_!) [below of=x_!] {{\huge $0$}};
  \node[state] (y_0) [below of=x_0] {{\huge $3$}};
  \node[state] (y_1) [below of=x_1] {{\huge $5$}};
  \node[state] (y_2) [below of=x_2] {{\huge $5$}};
  \node[state] (y_3) [below of=x_3] {};
  \node[state] (y_4) [below of=x_4] {};

  \node[] (x) [left of=x_!, node distance=0.6cm] {$x$};
  \node[] (y) [left of=y_!, node distance=0.6cm] {$y$};
  \node[] (tempo) [below of=y, node distance=0.6cm,  scale=0.7] {$time$};

  \node[] (-1) [below of=y_!, node distance=0.6cm, scale=0.7] {$-1$};
  \node[] (0) [below of=y_0, node distance=0.6cm, scale=0.7] {$0$};
  \node[] (1) [below of=y_1, node distance=0.6cm, scale=0.7] {$1$};
  \node[] (2) [below of=y_2, node distance=0.6cm, scale=0.7] {$2$};
  \node[] (3) [below of=y_3, node distance=0.6cm, scale=0.7] {$3$};
  \node[] (4) [below of=y_4, node distance=0.6cm, scale=0.7] {$4$};

  \node[] (pq) [above of=x_0, node distance=0.6cm] {$\{p,r,q\}$};
  \node[] (p) [above of=x_1, node distance=0.6cm] {$\{p\}$};
  \node[] (pq2) [above of=x_2, node distance=0.6cm] {$\{r,q\}$};

  \path[->] (x_0) edge [] node {{\tiny$<$}} (x_1)
            (y_0) edge [] node[above, rotate=24, pos=0.3] {{\tiny$\leq$}} (x_2)
            (x_2) edge [] node {{\tiny$<$}} (x_3)
            (y_2) edge [] node[above, rotate=24] {{\tiny$\leq$}} (x_4)

            (x_1) edge [densely dashed] node {{\tiny$<$}} (x_2)
            (y_1) edge [densely dashed] node[above, rotate=24, pos=0.3] {{\tiny$\leq$}} (x_3)

            (y_!) edge [] node[above, rotate=45, pos=0.3] {\tiny$+1$} (x_0)
            (x_0) edge [] node[above, rotate=90, pos=0.3] {\tiny$+2$} (y_0)
            (y_0) edge [] node[above, rotate=45, pos=0.3] {\tiny$+1$} (x_1)
            (x_1) edge [densely dashed] node[above, rotate=90, pos=0.3] {\tiny$+2$} (y_1)
            (x_2) edge [] node[below, rotate=90, pos=0.3] {\tiny$+2$} (y_2)
            (y_1) edge [densely dashed] node[above, rotate=45, pos=0.3] {\tiny$+1$} (x_2);

\end{tikzpicture}
\caption{Constraint graph of $\{p,r,q\}\{p\}\{r,q\}$.}
\label{constr_fig}
\end{figure}

So far, we neglected any initialization condition, solving a \emph{general}
satisfiability problem. 
If a formula is shown to be unsatisfiable, then there
is no prefix of an infinite model $\pi_\sigma$, of length
equal to $k$, satisfying the formula.

\section{Bounded Reachability Problem}\label{sec:brp}
%
This section studies the bounded satisfiability of CLTLB($\D$) formulae
by using a finite representation of infinite models.
It is then shown that this entails 
the satisfiability of the same formula with respect to $k$-\emph{partial}
$\D$-valuations.
Finally, the section introduces the  Bounded (existential)
Reachability Problem (BRP) for Kripke structures, showing that BRP also admits a complete 
procedure.

First, we need to define a bounded semantics, i.e., a semantics of a formula on finite structures.
Let $k>0$, let $\phi$ be a CLTLB($\D$) formula and let
$\widehat{\sigma}_k: \{i\in\Z\mid
\lfloor\phi \rfloor_x \leq i\leq k + \lceil \phi\rceil_x\} \times
  \{x\} \to \D$, for each $x \in V$, called a {\em local
sequence}, be a finite sequence of assignements to variables in $V$. Informally, sequence
$\widehat{\sigma}_k$
is not only defined between instants 0 and $k$, but it is bordered by two segments defining
variable values before $0$ and after $k$, as shown also in
Fig. \ref{constr_fig}. This is necessary to correctly define the value of all a.t.t's
in the interval from $0$ to $k$; in fact,
the evaluation of an a.t.t. may involve also a bounded number of instants before instant 0 or after
instant $k$. 
Let $\pi\in S^+$, called a finite path. A finite path is {\em cyclic} if it is of the form $usvs$,
for
some $s\in S$, $u,v\in S^*$. A cyclic finite path can be considered a finite representation of an
infinite one, e.g., $u(sv)^\omega$. If $\pi$ is a cyclic path $usvs$, then a bounded semantics for
$\phi$ over $\pi$ and local assignment $\widehat{\sigma}_k$ is defined as in the case of a
$k$-partial $\D$-valuation of Section~\ref{sec:dec}, by replacing $\sigma_k$ with
$\widehat{\sigma}_k$ and $\pi$ with $u(sv)^\omega$ in (\ref{eq:rel-over-partial}). 
If $\pi$ is not cyclic, instead, the semantics of each relation $R$ is, for $0 \leq i \leq k$:
\begin{equation*}
\begin{aligned}
 \pi_{\widehat{\sigma}_k}^i \models_k R(\varphi_1, \dots, \varphi_n) &\iFF \\
   (\widehat{\sigma}_k^{i+|\varphi_1|}(&x_{\varphi_1}), \dots,
   \widehat{\sigma}_k^{i+|\varphi_n|}(x_{\varphi_n}) ) \in R
\end{aligned}
\end{equation*}
The bounded semantics of temporal operators is the same as the one in \cite{BCCZ99}, e.g.:
\begin{equation*}
\begin{aligned}
 \pi_{\widehat{\sigma}_k}^i \models_k \phi\U\psi &\iFF
 \left\{ \begin{gathered}
   \exists \, i\leq j \leq k: \pi_{\widehat{\sigma}_k}^j \models \psi \, \wedge \\
\pi_{\widehat{\sigma}_k}^n
 \models \phi \quad \forall\, i\leq n < j
 \end{gathered}\right.
\end{aligned}
\end{equation*}
\begin{equation*}
\begin{aligned}
 \pi_{\widehat{\sigma}_k}^i \models_k \phi\R\psi &\iFF
 \left\{
 \begin{gathered}
 \exists \, i\leq j \leq k: \pi_{\widehat{\sigma}_k}^j \models_k \phi \, \wedge \\
 \pi_{\widehat{\sigma}_k}^n \models_k \psi \quad \forall\, i \leq n \leq j
 \end{gathered}\right.
\end{aligned}
\end{equation*}
\begin{equation*}
\begin{aligned}
 \pi_{\widehat{\sigma}_k}^i \models_k \X\phi &\iFF
 \begin{gathered}
   0\leq i+1 \leq k\land \pi_{\widehat{\sigma}_k}^{i+1} \models \phi
 \end{gathered}
\end{aligned}
\end{equation*}

By using the bounded semantics, the following theorem holds:
\begin{theorem}\label{th:BoundedSAT}
For every \emph{CLTLB($D$)} formula $\phi$, 
if, there exist $k>0$, a finite path $\pi$ of length $k$ and a local
assignment $\widehat{\sigma}_k$ such that $\pi_{\widehat{\sigma}_k} \models_k \phi$ then
$\phi$ is satisfiable over $k$-partial $\D$-valuations. 
\end{theorem}
\begin{proof}[Proof sketch.]
The statement is proven by means of a completion of the sequence
$\widehat{\sigma}_k$ satisfying property (\ref{eq:rel-over-partial}). 
A legal completion may also involve undefined values: constraints encompassed in the loop of
$\pi_{\widehat{\sigma}_k}$ can be suitably bordered.
In particular, if $\pi=uv^\omega$ and $l$ is the length of 
$v$, for each variable $x$ such that $\lceil \phi
\rceil_x>0$, $\forall c \in D$, then $\forall h \geq  0$, $(k+1+hl, x, c) \not \in \widehat{\sigma}_k$.
By exploiting the results in \cite{BCCZ99} and a syntactic rewriting of
each $\D$ constraint with a propositional letter, which results in a formula
$\phi'$, from $\phi$, satisfied by a propositional model $\pi'$, then
$\pi' \models_k \phi'$  implies $\pi' \models \phi'$.
\end{proof}
%


The above concepts can be generalized and extended in the case of $\D$-Kripke
structures,
as suggested in \cite{DD02}.
\begin{definition}
A \textup{$\D$-Kripke structure} is a tuple $M=\langle S, T, C, \lambda\rangle$ with a
finite set of states $S$, a transition relation $T \subseteq S \times
S$ between states, a set $C$ of $\D$ relations on a.t.t.'s and a labeling function $\lambda: S \to 2^{AP}\times C$.
\end{definition}
Given a $\D$-Kripke structure $M$, a CLTLB($\D$) formula $\phi$ and an
initial state $s_0$, the {\em existential} model checking (MC) problem amounts to checking if
there exists a linear structure $\pi_\sigma$ 
such that $\pi_\sigma \models \phi$.
Because of the undecidability results of Section~\ref{sec:undec}, the existential MC problem 
must be redefined for $k$-\emph{partial}
$\D$-valuations in order to have a decidable under-approximation.
Thanks to the well-known representation of Kripke structures through LTL
formulae, and by considering a.t.t.'s in $C$  
as atomic elements, it is possible to obtain a
CLTLB($\D$) formula $\chi_{\textup{M}}$ defining the ``propositional'' description
of the language of $\D$-Kripke structure $M$.
The {\em $k$-partial $\D$-evaluation} model checking problem is defined as the
satisfiability of $\chi_{\textup{M}} \wedge \phi$ over $k$-partial $\D$-evaluations.

Theorem~\ref{th:BoundedSAT} may be strengthened for $\D$-Kripke structures when $\phi$ is a
reachability formula. Formula $\phi$ is a {\em reachability} formula when it is of the
form ${\bf F}\psi$, where $\psi$ is a CLTLB($\D$) formula without temporal operators (which
are allowed only in a.t.t.). Then, the {\bf\em Bounded Reachability Problem} (BRP) for $M$ and
$\phi$ is defined as the existence of $k>0$, a finite path $\pi$ of length $k$ and a local
assignment $\widehat{\sigma}_k$ such that $\pi_{\widehat{\sigma}_k} \models_k \chi_{\textup{M}}
\wedge \phi$.

\begin{corollary}\label{cor:BMC}
For every reachability formula $\phi$ in \emph{CLTLB($\D$)} and for every  $\D$-Kripke
structure $M$, the BRP is equivalent to the $k$-partial $\D$-evaluation MC problem. 
\end{corollary}
%

\section{Encoding of the Bounded Reachability Problem}
\label{sec:encoding}

In this section the BRP is encoded as the
satisfiability of a quantifier-free formula in the theory
$\text{EUF}\cup \mathcal{D}$ (QF-UF$\mathcal{D}$), where EUF is the theory of Equality and Uninterpreted
Functions, provided that the set $D$ includes a copy of $\N$ and that
$\text{EUF}\cup \mathcal{D}$ is consistent. The last condition is
easily verified in the case of a union of two consistent, disjoint,
stably infinite theories (as is the case for EUF and arithmetic).
In \cite{BCFPR10} a similar encoding is described for the case of
Integer Difference Logic (DL) constraints: in that case it results to be more succinct and expressive than the
Boolean one: lengthy propositional constraints are substituted
by more concise DL constraints and arithmetic (infinite) domains do
not require an explicit finite representation.
These facts, considering also that the satisfiability problem for the
quantifier-free fragment of $\text{EUF}\cup\text{DL}$ (QF-UFIDL) has the same complexity of SAT, make this approach
particularly efficient, as
demonstrated by the tests outlined in Section \ref{sec:tests}.

Under the above assumption, the proposed encoding is an effective
proof of the decidability of the BRP over $k$-partial $\D$-valuations.
In the general case an estimation of the complexity of the
satisfiability problem (for quantifier-free formulae) can be performed
via the Nelson-Oppen Theorem \cite{O80} as shown in Corollary \ref{cor:compl}.

As discussed before, the BMC problem amounts to looking for a finite
representation of infinite (possibly periodic) paths.
The Boolean approach \cite{BCCZ99} encodes finite paths by means of $2k+3$
propositional variables, while the same temporal behavior can be defined by means of \emph{one} QF-UF$\mathcal{D}$
formula involving only \emph{one} \emph{loop-selecting} variable $\bm{loop} \in D$:
\[
\bigwedge_{i=1}^k \left( (\bm{loop}=i) \Rightarrow L(s_{i-1}) = L(s_k) \right ).
\]
If the value $i$ of variable
$\bm{loop}$ is between $1$ and $k$, then there exists a loop, and it starts at $i$; notice that the formula $\bm{loop}=i$ is well defined since $D$ contains a copy of $\N$.

To encode a.t.t.'s, an \textit{arithmetic formula function}, i.e., an uninterpreted function
$\bm{\tau}: D \to D$, is associated with each arithmetic temporal subterm of
$\Phi$.
Let $\tau$ be such a subterm, then the arithmetic
formula function associated with it (denoted by the same name but written
in boldface), is recursively
defined w.r.t. the sequence of valuations $\sigma$ as:
\[
\begin{array}{c|c}
    \tau & 0 \leq i \leq k
    \\
    \hline
    x & \bm{x}(i) = \sigma^i(x) \\
    \X\alpha & \,\,\bm{\tau}(i) = \bm{\alpha}(i+1)    \\
    \Y\alpha & \,\,\bm{\tau}(i) = \bm{\alpha}(i-1)   \\
\end{array}
\]
If $D$ includes a copy of $\Z$, this semantics is well-defined between $0$ and $k$ thanks to the
initialization function $I$, otherwise we need to consider a \emph{shifted} function $\overline{\sigma}$ such that $\overline{\sigma}(i,\cdot)=\sigma(i-\lfloor\phi\rfloor,\cdot)$.

The propositional encoding is based on the one presented in
\cite{BHJLS06}, which is modified to take also into account relations over
a.t.t.'s.
In the case of the Boolean encoding, the truth value of a PLTLB formula $\Phi$ is
defined w.r.t. the truth value of its subformulae. For each subformula $t$, a set of
Boolean variables $\{t_i\}_{0\leq i \leq k+1}$ is associated with it: if $t_i$ holds, then subformula
$t$ holds at instant $i$.
Instant $k+1$ is introduced to more easily represent the instant
in which the periodic behavior starts.
The truth value of a CLTLB($\mathcal{D}$) formula
$\Phi$ is defined in a similar way.
The QF-UF$\mathcal{D}$ encoding, however, associates with each subformula $\theta$ a
\textit{formula predicate} that is a unary uninterpreted predicate (denoted by the same name but written in boldface)
$\bm{\theta} \in \mathcal{P}(D)$.
When the subformula $\theta$ holds at instant $i$ then $\bm{\theta}(i)$ holds.
As the length of paths is fixed to $k+1$ and all paths start from $0$,
formula predicates are actually subsets of $\{0, \dots, k+1\}$.
Let $\theta$ be a subformula of $\Phi$, $\alpha_1,\dots \alpha_n$ be a.t.t.'s and $R$ be an $n$-ary relation in $\mathcal{D}$; formula predicate
$\bm{\theta}$ is recursively defined as:
\[
\begin{array}{c|c}
  \theta & 0 \leq i \leq k+1 \\
  \hline
  p & \,\,\qquad\bm{\theta}(i) \iFF p \in L(s_i) \\
  R(\alpha_1,\dots,\alpha_n) & \,\,\bm{\theta}(i) \iFF R(\bm{\alpha_1}(i),\ldots ,\bm{\alpha_n}(i))\\
  \neg \phi &  \bm{\theta}(i) \iFF \neg \bm{\phi}(i)\\
  \phi \wedge \psi & \,\,\bm{\theta}(i) \iFF \bm{\phi}(i)
\wedge \bm{\psi}(i)
\end{array}
\]

{\em Temporal subformulae constraints} define the basic temporal
behavior of future and past operators, by using their traditional
fixpoint characterizations:
\[
\begin{array}{c|c}
  \theta & 0 \leq i \leq k \\
  \hline
  \X\phi & \bm{\theta}(i) \iFF \bm{\phi}(i+1) \\
  \phi\U\psi & \bm{\theta}(i)\iFF(\bm{\psi}(i) \vee (\bm{\phi}(i) \wedge
\bm{\theta}(i+1)))\\
  \phi\R\psi & \bm{\theta}(i)\iFF(\bm{\psi}(i) \wedge (\bm{\phi}(i) \vee
\bm{\theta}(i+1)))\\
\end{array}
\]
The encoding for the past operators is analogous to that for future operators except for the instant 0, which must be treated separately (see \cite{BCFPR10}).

{\em Last state constraints} define an equivalence between
truth at point $k+1$ and that at the point indicated by the $\bm{loop}$ variable, since the instant $k+1$ is representative of the instant $\bm{loop}$ along periodic paths.
Otherwise, for non-periodic paths, truth values in $k+1$ are trivially false.
These constraints have a similar structure to the
corresponding Boolean ones, but here
they are defined by only \emph{one} constraint, for each
subformula $\theta$ of $\Phi$,  w.r.t. the variable $\bm{loop}$:
\[
  \begin{array}{l}
    \left(\bigwedge_{i=1}^k (\bm{loop}=i) \Rightarrow
      \left(\bm{\theta}(k+1) \iFF \bm{\theta}(i))\right) \right) \wedge \\
    \left(\left(\bigwedge_{i=1}^k \neg(\bm{loop}=i)\right) \Rightarrow (\neg
\bm{\theta}(k+1))\right).
  \end{array}
\]
Note that if a loop does not exist then the fixpoint semantics of
$\R$ is exactly the one defined over finite acyclic paths in Section \ref{sec:brp}.
To correctly define the semantics of $\U$ and $\R$, their
\emph{eventualities} have to be accounted for.
Briefly, if $\phi\U\psi$ holds at $i$, then $\psi$ eventually holds
in some $j\geq i$; if $\phi\R\psi$ does not hold at $i$, then $\psi$
eventually does not hold in some $j\geq i$.
Along finite paths of length $k$, eventualities must hold between $0$
and $k$.
If a loop exists, an eventuality may hold within the loop.
The original Boolean encoding introduces $k$ propositional variables
for each subformula $\theta$ of $\Phi$ of the form $\phi\U\psi$ or $\phi\R\psi$ (one for each $1 \leq i \leq k$), which represent the eventuality of $\psi$
implicit in the formula, as first defined in \cite{BHJLS06}.
Instead, in the QF-UF$\mathcal{D}$ encoding, only \emph{one} variable $\bm{j_\psi} \in
D$
is introduced for each $\psi$ occurring in a subformula
$\phi\U\psi$ or $\phi\R\psi$.
\[
\begin{array}{c|c}
    \theta & \mathrm{Base} \\
    \hline
    \phi\U\psi &
    \begin{array}{l}
    \left(\bigvee_{i=1}^k \bm{loop}=i\right) \Rightarrow\\
    \qquad \qquad (\bm{\theta}(k) \Rightarrow
     \bm{loop} \leq \bm{j_\psi} \leq k \wedge \bm{\psi}(\bm{j_\psi}))
    \end{array}
    \\
    \phi\R\psi &
    \begin{array}{l}
    \left(\bigvee_{i=1}^k \bm{loop}=i\right) \Rightarrow\\
     \qquad (\neg\bm{\theta}(k) \Rightarrow
     \bm{loop} \leq \bm{j_\psi} \leq k \wedge \neg \bm{\psi}(\bm{j_\psi}))
    \end{array}
  \end{array}
\]
The complete encoding of $\Phi$ consists of the logical
conjunction of all constraints above, together with $\Phi$ evaluated at the
first instant along the time structure.

If $m$ is the total number of subformulae and $n$ is the total number of
temporal operators $\U$ and $\R$ occurring in $\Phi$, then the Boolean
encoding requires $(2k+3) + (k+2)m + (k+1)n = O(k(m+n))$ fresh propositional
variables.
The QF-UF$\D$ encoding
requires only $n+1$ integer variables ($\bm{loop}$ and $\bm{j_\psi}$) and $m$ unary
predicates (one for each subformula).


As previously anticipated, if $\D$ is a consistent, stably infinite theory, $\phi$ is a
formula of length $n$ and $T(n)$ is the
complexity of the satisfiability problem in $\D$ then, by the Nelson-Oppen
Theorem, the satisfiability of a CLTLB($\D$) formula $\phi$
over $k$-partial $\D$-valuations can be solved in $O(2^{n^2}(nk\log{(nk)} +
T(nk)))$; moreover, if $\D$ is convex it can be solved in $O(n^3(nk\log{(nk)} + T(nk)))$.

\begin{corollary}\label{cor:compl}
The satisfiability of a CLTLB($\D$) formula over $k$-partial
$\D$-valuations is \textsc{NP}-complete when $\D$ is
DL, \textsc{P} when $\D$ is RDL (Real DL) and $4$-\textsc{EXPTIME}
when $\D$ is LIA (Linear Integer Arithmetic).
\end{corollary}

\section{Experimental Results} \label{sec:tests}

The encoding presented in Section \ref{sec:encoding} for CLTLB(DL) has been implemented as a
plugin of the Zot tool\footnote{Zot is available at {\em home.dei.polimi.it/pradella}.}.
This implementation exploits SMT solvers as verification engines, and in
particular it is based on the SMT-LIB \cite{RT06} to achieve independence from
the particular SMT solver used\footnote{As SMT solvers we used both Yices
({\em yices.csl.sri.com}) and Z3
 ({\em research.microsoft.com/\-en-us/\-um/\-redmond/\-projects/z3)}.}.
The Zot plugin has been used to carry out a number of experiments on a variety of examples, old and new.
For the sake of brevity, we do not report here the full experimental
data\footnote{The data are available at {\em home.dei.polimi.it/bersani}.}, and we only briefly summarize them in an informal way.

We carried out two kinds of experiments.
First, we used the new encoding to perform BMC on a set of previously defined PLTLB specifications, to compare the performances of the new Zot plugin w.r.t. the existing SAT-based one presented in \cite{PMS09}.
The SMT-based encoding showed considerable improvements in the vast majority of experiments, for both of the SMT solvers used. The recorded speedup (computed as the ratio $T_{SAT}/T_{SMT}$) was always substantial, and in many cases it was more than 
tenfold (often considerably more than that). For example, we repeated the experiments of \cite{BFPR09} with the new encoding, and the average speedup in the overall verification time was around 2.4 with Z3, and 21.4 with Yices; 
we point out that the gains in performance were particularly significant for the most complex specifications. 

In the second set of experiments we exploited also the new features of CLTLB(DL) w.r.t. PLTLTB, and we used the bounded reachability results presented in Section \ref{sec:brp} to analyze some relevant aspects of non-trivial applications based on the Service-Oriented paradigm \cite{BCFPR10}.
On examples that fall in the range of properties expressible through both CLTLB(DL) and 
PLTLB (e.g., those that involve only bounded domains), the performances of the SMT-based verification are, again, an order of magnitude better than the SAT-based one (the average performance speedup over such properties was 55 with Z3 and 7.4 with Yices).


\section{Conclusions and Future Work} \label{sec:conclusions}
In this paper, we introduced the logic CLTLB($\D$), an extension of PLTLB allowing
as subformulae arithmetic
constraints belonging to a generic constraint system $\D$.
We introduced suitable assumptions
concerning the structure of models, to make satisfiability 
of CLTLB($\D$) decidable,
provided that $\D$ has, in turn, a decidable decision procedure.
In this case, the Bounded Reachability Problem (BRP) for CLTLB($\D$)
formulae can be solved by means of automatic software verification tools.
We built a Bounded Reachability Checker by using SMT-solvers which natively implement 
decision procedures for QF-UF$\D$ when $\D$ is DL or LIA, with very encouraging experimental results.

Future work will compare the new
arithmetic-based encoding with existing Boolean ones by means of a
comprehensive set of tests; we also intend to define
new extensions representing infinite behaviors of variables 
and search for suitable classes of formulae
inducing actual $\omega$-periodic models.
\subsubsection*{Acknowledgments}
Many thanks to Luca Cavallaro for providing stimulating case studies.
This research has been partially funded by the European Commission,
Programme IDEAS-ERC, Project 227977-SMScom, and by the Italian
Government under the project PRIN 2007 D-ASAP (2007XKEHFA).


\bibliographystyle{latex8}
\bibliography{time10-ea-bib}

\begin{thebibliography}{10}\setlength{\itemsep}{-1ex}\small

\bibitem{BCFPR10}
M.~M. Bersani, L.~Cavallaro, A.~Frigeri, M.~Pradella, and M.~Rossi.
\newblock {SMT-based Verification of LTL Specifications with Integer
  Constraints and its Applications to Runtime Checking of Service
  Substitutability}.
\newblock Technical report, arXiv:1004.2873v1, 2010.

\bibitem{BFPR09}
M.~M. Bersani, C.~A. Furia, M.~Pradella, and M.~Rossi.
\newblock Integrated modeling and verification of real-time systems through
  multiple paradigms.
\newblock In {\em Proc. of SEFM}, pages 13--22, 2009.

\bibitem{BCCZ99}
A.~Biere, A.~Cimatti, E.~M. Clarke, and Y.~Zhu.
\newblock Symbolic model checking without {BDDs}.
\newblock In {\em Proc. of TACAS}, pages 193--207, 1999.

\bibitem{BHJLS06}
A.~Biere, K.~Heljanko, T.~A. Junttila, T.~Latvala, and V.~Schuppan.
\newblock {Linear Encodings of Bounded LTL Model Checking}.
\newblock {\em Log. Meth. in Comp. Sci.}, 2(5), 2006.

\bibitem{B98}
B.~Boigelot.
\newblock {\em Symbolic Methods for Exploring Infinite State Spaces}.
\newblock PhD thesis, Universit\'{e} de Li\`{e}ge, 1998.

\bibitem{CC00}
H.~Comon and V.~Cortier.
\newblock {Flatness Is Not a Weakness}.
\newblock In {\em CSL}, pages 262--276, 2000.

\bibitem{CJ98}
H.~Comon and Y.~Jurski.
\newblock {Multiple Counters Automata, Safety Analysis and Presburger
  Arithmetic}.
\newblock In {\em CAV}, pages 268--279, 1998.

\bibitem{dMRS02}
L.~M. de~Moura, H.~Rue{\ss}, and M.~Sorea.
\newblock Lazy theorem proving for bounded model checking over infinite
  domains.
\newblock In {\em CADE}, pages 438--455, 2002.

\bibitem{D04}
S.~Demri.
\newblock {LTL over Integer Periodicity Constraints: (Extended Abstract)}.
\newblock In {\em FoSSaCS}, pages 121--135, 2004.

\bibitem{DD02}
S.~Demri and D.~D'Souza.
\newblock An automata-theoretic approach to constraint {LTL}.
\newblock In {\em FSTTCS}, pages 121--132, 2002.

\bibitem{DFGD06}
S.~Demri, A.~Finkel, V.~Goranko, and G.~van Drimmelen.
\newblock {Towards a Model-Checker for Counter Systems}.
\newblock In {\em ATVA}, pages 493--507, 2006.

\bibitem{DG06}
S.~Demri and R.~Gascon.
\newblock {The Effects of Bounding Syntactic Resources on Presburger LTL}.
\newblock Technical Report {LSV-06-5}, LSV, 2006.

\bibitem{E90}
E.~A. Emerson.
\newblock Temporal and modal logic.
\newblock In {\em Handbook of Theoretical Computer Science, Volume B: Formal
  Models and Sematics (B)}, pages 995--1072. 1990.

\bibitem{G87}
D.~M. Gabbay.
\newblock The declarative past and imperative future: Executable temporal logic
  for interactive systems.
\newblock In {\em Temporal Logic in Specification}, pages 409--448, 1987.

\bibitem{GPSS80}
D.~M. Gabbay, A.~Pnueli, S.~Shelah, and J.~Stavi.
\newblock On the temporal basis of fairness.
\newblock In {\em POPL}, pages 163--173, 1980.

\bibitem{Kam68}
J.~A.~W. Kamp.
\newblock {\em Tense Logic and the Theory of Linear Order}.
\newblock PhD thesis, University of California at Los Angeles, 1968.

\bibitem{O80}
D.~C. Oppen.
\newblock Complexity, convexity and combinations of theories.
\newblock {\em Theor. Comput. Sci.}, 12:291--302, 1980.

\bibitem{PMS07}
M.~Pradella, A.~Morzenti, and P.~{San Pietro}.
\newblock The symmetry of the past and of the future: bi-infinite time in the
  verification of temporal properties.
\newblock In {\em ESEC/SIGSOFT FSE}, pages 312--320, 2007.

\bibitem{PMP08}
M.~Pradella, A.~Morzenti, and P.~{San Pietro}.
\newblock Refining real-time system specifications through bounded model- and
  satisfiability-checking.
\newblock In {\em ASE}, pages 119--127, 2008.

\bibitem{PMS09}
M.~Pradella, A.~Morzenti, and P.~{San Pietro}.
\newblock A metric encoding for bounded model checking.
\newblock In A.~Cavalcanti and D.~Dams, editors, {\em FM 2009: Formal Methods},
  volume 5850 of {\em LNCS}, pages 741--756. Springer, 2009.

\bibitem{RT06}
S.~Ranise and C.~Tinelli.
\newblock The {SMT-LIB} standard: Version 1.2.
\newblock Technical report, 2006.
\newblock http://combination.cs.uiowa.edu/smtlib/.

\bibitem{Sch02}
P.~Schnoebelen.
\newblock The complexity of temporal logic model checking.
\newblock In {\em Adv. in Modal Logic}, pages 393--436, 2002.

\end{thebibliography}

\end{document}